\documentclass[a4paper,UKenglish,cleveref, autoref, thm-restate]{lipics-v2021}

\title{Continuous Map Matching to Paths\\ under Travel Time Constraints}

\titlerunning{Continuous Map Matching to Paths}

\author{Yannick Bosch}{University of Konstanz, Germany}{yannick.bosch@uni-konstanz.de}{https://orcid.org/0009-0005-8517-9226}{}

\author{Sabine Storandt}{University of Konstanz, Germany}{sabine.storandt@uni-konstanz.de}{https://orcid.org/0000-0001-5411-3834}{}

\authorrunning{Y. Bosch and S. Storandt} 



\category{} 

\relatedversion{} 

\nolinenumbers

\begin{document}

\maketitle

\begin{abstract}
 In this paper, we study the problem of map matching with travel time constraints. Given a sequence of $k$  spatio-temporal measurements and an embedded path graph with travel time costs, the goal is to snap each measurement to a close-by location in the graph, such that  consecutive locations can be reached  from one another along the path  within the timestamp difference of the respective measurements.
This problem arises in public transit data processing as well as in map matching of movement trajectories to general graphs.
We show that the classical approach for this problem, which relies on selecting a finite set of candidate locations in the graph for each measurement, cannot guarantee to find a consistent solution. We propose a new algorithm that can deal with an infinite set of candidate locations per measurement. We prove that our algorithm  always detects a consistent map matching path (if one exists). Despite the enlarged candidate set, we also demonstrate that our algorithm has superior running time in theory and practice.
For a path graph with $n$ nodes, we show that our algorithm runs in  $\mathcal{O}(k^2 n \log {nk})$ and under mild assumptions in $\mathcal{O}(k n ^\lambda + n  \log^3 n)$ for $\lambda \approx 0.695$. This is a significant improvement over the baseline, which runs in  $\mathcal{O}(k n^2)$ and which might not even identify a correct solution. The performance of our algorithm hinges on an efficient segment-circle intersection data structure. We describe how to design and implement such a data structure for our application. In the experimental evaluation, we demonstrate the usefulness of our novel algorithm on a diverse set of generated measurements  as well as GTFS data.
\end{abstract}
%\vspace{1.5em}
\section{Introduction}
Map matching is the process of aligning a sequence of uncertain position measurements  to a path in a given network, such that movement along the paths explains the observed measurements well. Oftentimes, the measurements also contain temporal information, which should be incorporated in the mapping process to obtain a meaningful match. There are many applications of map matching, including real-time navigation \cite{bernstein1996introduction,white2000some}, traffic analysis \cite{giovannini2011novel,goh2012online}, location-based services \cite{velaga2010development}, behavior analysis of animals and humans \cite{cho2017basis,edelhoff2016path}, as well as indexing and compression of large trajectory data  sets \cite{kellaris2013map,funke2019pathfinder}.

Formally, the input is defined as a sequence of $k$ measurements  $M = M_1, \dots, M_k$, where $M_i = (p_i, t_i)$ consists of the measured position $p_i=(x_i,y_i) \in \mathbb{R}^2$ and the timestamp $t_i \in \mathbb{R}$ of the measurement. Further given is an embedded directed graph $G(V,E)$ with node coordinates $(v_x,v_y) \in \mathbb{R}^2$ for each node $v \in V$. Edges are embedded as straight line segments and are augmented with travel time costs $t : E \rightarrow \mathbb{R}^+$. Edge costs are interpolated linearly between the end points of the edge.
We say a location $l=(x,y)$ \emph{exists} in $G$ if there is an edge $e = \{v,w\}\in E$ such that  $l = \lambda v + (1-\lambda) w$ for some $\lambda \in [0,1]$. That means, $l$ is a point on the straight line segment that represents $e$. The goal of map matching is to assign each measurement $M_i$ to a location $l_i$ that exists in $G$, such that $l_i$ is in close vicinity to $p_i$ and $l_{i+1}$ can be reached from $l_i$ on a path in $G$ whose travel time does not exceed the temporal difference between consecutive measurements, that is  $t_{i+1}-t_i$. The concatenation of such paths for $i=1, \dots, k-1$ forms the overall \emph{matching path}.

In this paper, we focus on the scenario where $G$ is a path graph. 
Map matching to paths arises as a subproblem in map matching to general graphs in approaches that first extract a set of candidate paths and then check their feasibility as a matching path. This kind of approach was used, for example, in \cite{custers2022physically} where the goal is to find a physically consistent matching path (which obeys maximum road speeds as well as constraints on acceleration and deceleration). Furthermore, the problem also plays a role in  the processing of public transit data. In the so called bus stop mapping problem, the path graph is defined by a bus line and the measurement sequence consists of bus stop positions together with bus departure times. Oftentimes, the bus stop positions are not directly on the bus line or there are multiple close-by bus stops and thus they first need to be matched correctly to allow for public transit routing and other applications \cite{geops_snapping_stops}. Figure \ref{fig:Olbia} shows an example instance.
\begin{figure}
    \centering
    \includegraphics[width=\textwidth]{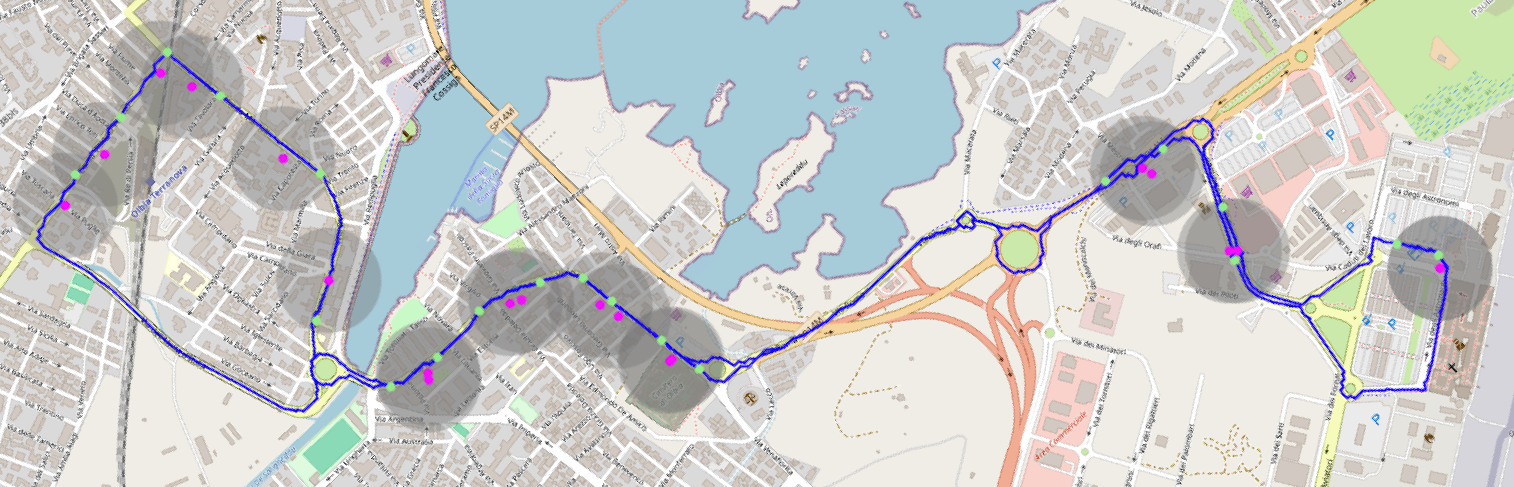}
    \caption{Bus stop mapping for a circular bus route in Olbia, Italy. The radius is chosen as 110 meters for visual clarity. Marked in   light green is the mapped sequence.}
    \label{fig:Olbia}
\end{figure}

\subsection{Issues with Existing Methods}
A naive attempt for map matching is to assign  each measurement $M_i$ to the closest node in $G$ or to the closest point on an edge in $G$ \cite{white2000some}. However, this method does not produce sensible matching paths reliably as it neglects the necessary continuity of the movement in the graph \cite{krumm2007map}. A more sophisticated and wide-spread method for map matching is the following one \cite{lou2009map,eisner2011algorithms,hu2016if}, which we refer to as DAG  algorithm:
\begin{itemize}
    \item For each $M_i$, select a finite set of candidate locations $L_i$ that exist in $G$. 
    \item Construct a layered DAG where layer $i$ contains a node for each element in $L_i$ and where there are edges $(a,b)$ for each $a \in L_i, b \in L_{i+1}$. Edge costs correspond to the minimum travel time between the respective nodes in the graph.
    \item Prune edges that do not obey given travel time constraints between consecutive locations.
    \item Find the shortest path $\pi$ from a node in layer $1$ to a node in layer $k$ and return it.
\end{itemize}
The advantage of this approach is its simplicity. However, the drawback of restricting candidate locations to a finite candidate location set  in the underlying graph is that the quality of the returned path and the efficiency of the method crucially  depend on the  density and distribution of the candidate locations.
Most commonly, $L_i$ is either the set of the $c$-closest nodes to $p_i$ in $V$ for some $c \in \mathbb{N}$ or the set of nodes within a disk $D_i$ of radius $r$ centered at $p_i$ (where $r$ models the  location measurement uncertainty).

Potential issues that arise when using these methods in the DAG approach are illustrated in Figure \ref{fig:mm}.
\begin{figure}
    \centering
    \includegraphics[width=0.6\linewidth]{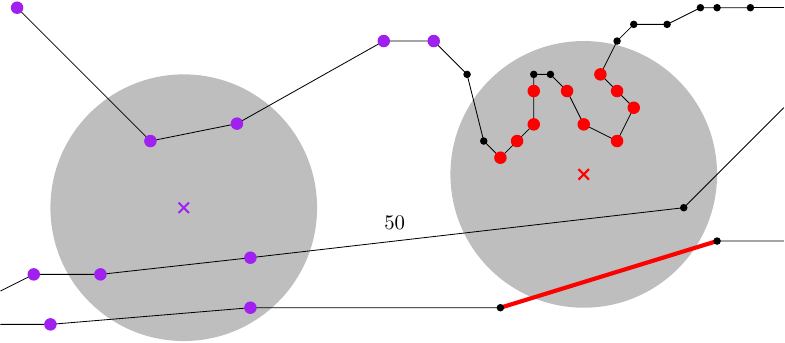}
    \caption{Potential issues for map matching when using finite candidate location sets: The $\times$ symbols mark measured positions. For the right measurement, its $c=10$ closest nodes in the graph, marked in red, are all on one road and other close-by roads are thus not considered. For the left measurement, its $c=10$ closest nodes are quite far spread and not necessarily close to the measured position. If all nodes in a disk around the measurement are considered (indicated by the gray areas), positions on segments as the one marked in  red are ignored as none of its endpoints is inside the respective disk. In the middle part of the image,  a road segment of cost 50 is shown. If the allowed travel time from the purple to the red measurement is 40, then this road section could not be part of a feasible matching path. However, there is clearly a position on the road segment within the disk around the red measurement that can be reached after 40 units of travel time (using linear interpolation). Thus, excluding this segment results in a false negative.}
    \label{fig:mm}
\end{figure}
If the candidate set is formed by the $c$ nodes that are closest to the measured location, selecting a sensible value for $c$ is challenging. If there are densely sampled paths in the graph (which often is the case for curvy routes in real networks), then even with a large value of $c$ all candidate locations might be on one road; and good alternatives, which are slightly further away, could be completely missed. Setting $c$ to a large value might result in a mapping to locations that are far away from the measured position. It also increases the number of nodes in the constructed DAG and with that the running time of the approach.
When using disks $D_i$ of radius $r$ to determine the candidate locations as $D_i \cap V$, feasible mapping locations can be  missed because $D_i$  might only intersect an edge but contains none of its endpoints. Furthermore, the running time for computing the DAG might also become huge if many nodes are contained in the disks, as shortest paths have to be computed between all pairs of candidate locations for consecutive layers.
In general, using only a fixed set of candidate locations becomes particularly problematic  when travel time constraints are taken into account. It might be the case  that there exists a path $\pi_{i,i+1}$ in $G$ that connects a location close to $p_i$ to a location close to $p_{i+1}$ with a travel time of at most $t_{i+1}-t_i$. However, if the respective start and end locations are on edges and not on nodes, it might be that only a superpath of $\pi_{i,i+1}$ is represented in the DAG with a travel time that exceeds  $t_{i+1}-t_i$. Thus, the algorithm cannot identify a feasible matching path even though there is one.  Clearly, all of the discussed issues can already arise if $G$ is a path graph. And even for this special case, the described DAG  algorithm is the prevailing standard \cite{geops_snapping_stops}.

In this paper, we propose a continuous map matching approach that does not require to fix a  finite set of candidate locations for each measurement a priori. Instead, we allow $M_i$ to be mapped to any location in $D_i$ that exists in $G$.  This prohibits the usage of the DAG  algorithm,  since there can now be an infinite number of candidate locations in each layer. We prove that our novel algorithm always identifies a feasible matching path (if one exists) and that it does so  even faster in practice  than the DAG algorithm which might fail to produce a correct result due the its restricted candidate location set.

\subsection{Further Related Work}
Many approaches for map matching have been proposed in the literature, using different input models, objective functions, and assumptions about the underlying embedded graph.

In \cite{lou2009map,custers2022physically,hu2016if}, the DAG algorithm is used with a more refined candidate selection strategy. Here, first the $c$ closest edges to the measured position or the edges within a certain radius  are selected and then a suitable point per edge is chosen. However, the potential issues described above for finite candidate location sets also apply to these strategies. In \cite{eisner2011algorithms}, engineering methods are described to accelerate the baseline algorithm with  candidate locations being selected as nodes in disks of radius $r$. In particular, methods for fast one-to-many shortest path computations are shown to reduce the running time significantly. We will show that for path graphs, even faster solutions are possible.

A completely different approach for map matching is to interpret $M$ as a polyline by linear interpolation between the measured positions and to find the path in $G$ with minimum Fr\'echet distance to that polyline \cite{gudmundsson2024map,chambers2020map}. With this objective, there is no need to specify any restricted candidate location sets and one gets a continuous mapping sequence. However, the linear interpolation between position measurements might not be realistic, especially if measurements are sparse. Moreover, the computation of the exact Fr\'echet distance is computationally expensive. Thus, for fast query times, one has to resort to approximations \cite{chen2011approximate}. Finally, it is unclear how to incorporate travel time constraints into this framework.  Finally, as the Fr\'echet distance is a bottleneck measure,  there is usually a large set of matching paths with the same objective function value. Thus, oftentimes there needs to be a second map matching step with another objective function or a combined objective to select the final matching path \cite{wei2013map}. 
Other polyline distance measures, as the weak Fr\'echet distance or the vertex-monotone Fr\'echet distance \cite{brakatsoulas2005map,wei2013map} as well as dynamic time warping \cite{wakuda2012adaptive}, have also been applied in the map matching context. But they all suffer from similar issues and, in particular, the respective methods are not designed to identify a feasible matching path with respect to time constraints.

There is also a plethora of HMM or machine learning based methods for map matching \cite{krumm2007map,mohamed2016accurate,hashemi2016machine,huang2021survey,zhao2019deepmm}. However, those also do not come with quality guarantees and they either require a fine tuned model or high quality training data to work well in practice.

In \cite{bast2018sparse,vuurstaek2020gtfs}, map matching of bus stop sequences to a transport network are discussed. They also leverage the DAG algorithm but incorporate additional requirements as e.g turn restrictions or consistency among different bus lines. However, travel time constraints are also not taken into account.

\section{Contribution}
We study the following map matching problem: Given a spatio-temporal measurement sequence $M$ of size $k$, an embedded directed path $P$ with travel time costs $t$, as well as $r \in \mathbb{R}_{\geq 0}$, assign to each measurement $M_i=(p_i,t_i)$ a location $l_i$ that exists in $P$ such that
\begin{itemize}
    \item $\forall i \in [k]: d(p_i,l_i) \leq r$, that is, $l_i$ lies in the disk $D_i$ of radius $r$ centered at  $p_i$,
    \item $\forall i=2,\dots,k: t(l_{i-1},l_i) \leq t_i-t_{i-1}$, that is, the travel time in $P$ between consecutive locations does not exceed the time difference of the respective measurements.
\end{itemize}
 We call a mapping sequence $l_1, \dots,l_k$ feasible if it fulfills these constraints.

This basic formulation is of high practical relevance. Nevertheless, we are not aware of an algorithm that is guaranteed to find a feasible sequence of locations if one exists. While DAG based approaches often identify good solutions in practice, they might fail due to the restricted candidate location sets. We propose a novel COntinuous Map Matching Algorithm (COMMA),  which does not rely on  finite  candidate location sets but instead allows $M_i$ to be mapped to any location in $D_i \cap P$. Consequently,  the DAG based approach can no longer be applied as every layer in the DAG might now contain an infinite number of nodes.  Our main insight is that we can compute a feasible matching path based on detecting its intersections with the boundaries of the disks $D_1, \dots, D_k$.
Using this idea, we can fully avoid the construction of a DAG which significantly reduces the running time.
We show that COMMA can find a feasible mapping sequence (if one exists)  in $\mathcal{O}(k n \log n)$ for non-overlapping disks and in $\mathcal{O}(k^2n \log {nk})$ for general inputs. For $k \ll n$, this is a significant reduction compared to the $\mathcal{O}(k n^2)$ running time of the DAG based approach  (which might not find a feasible solution). Under realistic assumptions about the input, the running time of  COMMA can be further reduced to $\mathcal{O}(k n^{0.695} + n \log^3n)$ or $\mathcal{O}(k \sqrt n \log^2n + n\sqrt n \log^2n)$ with the help of suitable segment-circle intersection reporting data structures. We describe how to design such a data structure for general segment sets and also how to tailor it to our special case in which the segments are known to form a path.

 We demonstrate the efficiency of the overall algorithm on generated measurement sequences as well as real bus stop mapping instances extracted from GTFS\footnote{\url{https://gtfs.org/}} data.

\section{COMMA for Path Graphs}
In this section, we describe our new map matching algorithm that computes a feasible mapping location $l_i$ for each measurement $M_i$ without constructing a DAG. Throughout the paper, we use the terms path edges and segments interchangeably. We refer to a subpath of $P$ from a point $p \in P$ to $p' \in P$ as an interval of $P$ defined by the travel time from the start of $P$ to $p$ and $p'$, respectively.  
 Our algorithm proceeds in two steps. It first identifies for each $M_i$ a set of feasible intervals in which the location $l_i$ can reside. In the  second step it extracts a feasible mapping sequence for $M$ as a whole (if one exists). 
For ease of exposition, we first consider the case that the disks $D_i$ of radius $r$ are pairwise non-overlapping and later discuss which modifications are needed in case this assumption is not met.

\subsection{Algorithm} The  following algorithm identifies for each $M_i$ its possible mapping locations that are part of some feasible matching path. The idea is to start with all intervals in $D_i \cap P$ and to then  exclude points that do not obey the travel time constraints with respect to $M_{i-1}$ or $M_{i+1}$.
\begin{enumerate}
    \item For each $i \in [k]$, compute the  intervals $I_i$ of $P$ inside the disk $D_i$ and sort them increasingly.
    \item Forward sweep: For $i=2,\dots,k$, consider the intervals in $I_i$. For each interval $[a,b]$, find the latest endpoint $b'$ of an interval in $I_{i-1}$ which is prior to $a$. If no such $b'$ exists, delete the interval. Set $b = \min(b,b'+t_i-t_{i-1})$. If $[a,b]=\emptyset$, delete the interval.
    \item Backward sweep:  For $i=k-1,\dots,1$, consider the intervals in $I_i$. For each interval $[a,b]$, find the earliest starting point $a'$ of an interval in $I_{i+1}$ which starts after $b$. If no such $a'$ exists, delete the interval. Set $a = \max(a,a'-t_{i+1}+t_{i})$. If $[a,b]=\emptyset$, delete the interval.
\end{enumerate}
The three steps are illustrated  in Figure \ref{fig:sweep}. 
The correctness  of the algorithm is shown below.
\begin{figure}
    \centering
    \includegraphics[width=\linewidth]{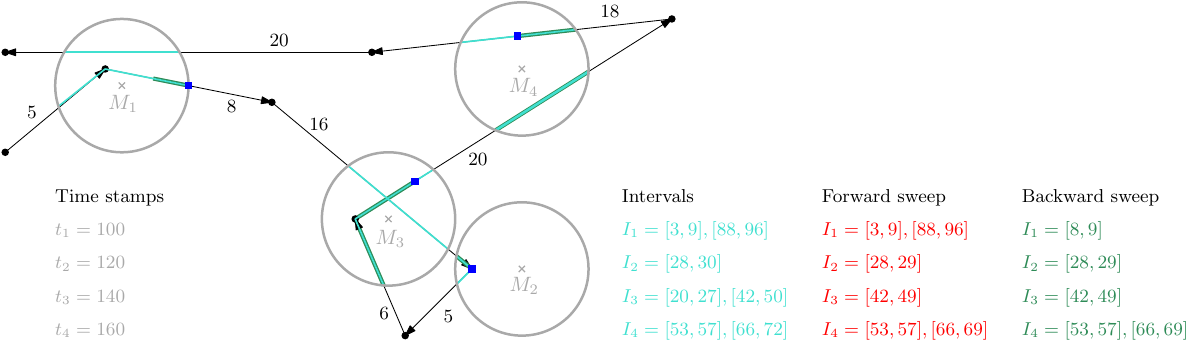}
    \caption{Example execution of COMMA. The underlying path $P$ along with its travel time costs is shown in black. A sequence of four measurements is indicated by the measured positions ($\times$) and the listed time stamps (in gray). The initial intervals that result from intersecting the disks around the measured positions with $P$ are provided in turquoise. The modified intervals after both sweeps are shown in green. The blue squares mark a feasible location mapping  sequence.}
    \label{fig:sweep}
\end{figure}
\begin{lemma}\label{lem:sweep}
    Let $I_i, i \in [k]$ denote the  final intervals after the two sweeps, then for each interval $[a,b] \in I_i, i=1, \dots, k-1$ and for each location $l \in [a,b]$, there exists  $[a',b'] \in I_{i+1}$ and $l' \in [a',b']$ such that the travel time from $l$ to $l'$ on $P$ is at most $t_{i+1}-t_i$, and vice versa.
\end{lemma}
\begin{proof}
    Assume for contradiction that there exists  $l \in [a,b]$ but no feasible $l'$. This implies that all intervals in $[a',b'] \in I_{i+1}$ that start after $l$ have $a' > l+t_{i+1}-t_i$. This applies in particular to the interval with smallest starting point greater than $b$. Thus, as the interval $[a,b]$ was considered in the backward sweep, the value of $a$ was set to $a' - t_{i+1} + t_i > l$. This contradicts the assumption that $l \in [a,b]$ as we have shown $a > l$. The same argumentation applies if the roles of $l'$ and $l$ are reversed by means of the forward sweep.
\end{proof}
If any $I_i$ becomes empty during the course of the algorithm, then the process is immediately aborted and it is certified that there is no sequence of mapping locations that result in a feasible matching path. Otherwise, we can retrieve a feasible location sequence  as follows:
\begin{itemize}
    \item Pick any location in $\biguplus_{I \in I_1}I$ as $l_1$.
    \item For $i=2,\dots,k$ pick $l_i \in \biguplus_{I \in I_i}I \cap [l_{i-1},l_{i-1}+t_i-t_{i-1}]$. Such a location always exists as shown in  Lemma \ref{lem:sweep} and by construction the travel time constraint is obeyed.
\end{itemize}
We remark that the backward sweep alone would suffice for this algorithm to succeed. However, combined with the forward sweep, the final intervals consist solely of locations that are part of some feasible mapping sequence. This is beneficial if secondary optimization goals, as for example, the summed distance between measured positions and mapped location shall be considered to select one of the feasible solutions. Figure \ref{fig:sweep} shows the location sequence that results from always picking the largest possible value in the respective intervals. The following lemma shows that this retrieval strategy is very efficient.
\begin{lemma}
    Given the set of feasible intervals $I_i, i \in [k]$, a feasible mapping sequence can be computed in $\mathcal{O}(k \log n)$.
\end{lemma}
\begin{proof}
    As the intervals in $I_i$ are pairwise non overlapping and sorted, $l_1$ can be set to the endpoint of the last interval in $I_1$. For $i = 2, \dots, k$, we find the last interval $[a, b] \in I_i$ with $a \leq l_{i-1} + t_i - t_{i-1}$ via binary search. Setting $l_i = \min(l_{i-1} + t_i - t_{i-1},b)$ yields the next location. Thus a mapping sequence can be found in $\mathcal{O}(k \log n)$.
\end{proof}

To assess the running time of the overall algorithm, we now describe how to conduct its steps in more detail. To get the intervals $I_i$ for each measurement, we traverse the path $P$ from start to end edge-by-edge, always keeping track of the total travel time on $P$ since the start, and check for intersection(s) of each edge segment with the boundary of  $D_i$. For each intersection, we compute its timestamp along $P$ and store these timestamps in sorted order. Note that we also include the start and the end of $P$ in this sorted list, in case they are contained in $D_i$. The result is a list of timestamps $s_1,s_2, \dots$ of even size, where odd entries mark interval starting points and even entries interval ending points. As each edge can have at most two intersection points with a disk boundary, we have $|I_i| \leq n$. Thus, traversing the path and computing the timestamps and intervals can be accomplished in time $\mathcal{O}(n)$ per measurement.
 In the forward sweep, we find the intervals in $I_{i-1}$ that are relevant for the right bound of the intervals in $I_i$ via binary search (and vice versa in the backward sweep) in $\mathcal{O}(\log n)$. The respective interval updates cost $\mathcal{O}(1)$ per interval. Thus, each sweep takes  $\mathcal{O}(k n \log n)$ in total. The location sequence  retrieval step can be done in $\mathcal{O}(k \log n)$. Thus, the two sweeps dominate the COMMA running time.

\subsection{Dealing with overlapping disks}
If we get rid of the assumption of non-overlapping disks, only intervals within one disk remain non-overlapping while intervals in $I_i$ and $I_{i+1}$ might now indeed share points. This impacts the interval update procedure during the sweeps. In the forward sweep, for an interval $[a,b] \in I_i$ we can no longer only consider the interval in $I_{i-1}$ with the latest endpoint prior to $a$. Instead, we still find this interval via binary search but then proceed to check the successive intervals in the sorted list for intersection with $[a,b]$ until we reach the end of the list or the first interval that starts after $b$. We then subdivide $[a,b]$ at the starting points of its overlapping intervals. For each subinterval $[a_s,b_s]$ created in this way that starts with an overlap with some interval $[a',b']$, we set $b_s$ to $\min(b_s,b'+t_i-t_{i-1})$. Figure \ref{fig:overlap} illustrates the process. Only the first subinterval might not have an overlapping portion  with an  interval from $I_{i-1}$. In this case, we can simply handle it with the procedure described above for non-overlapping intervals. If intervals in $I_i$ overlap after this process, they can be merged into one. The backward sweep works analogue.
\begin{figure}
    \centering
    \includegraphics[width=0.65\linewidth]{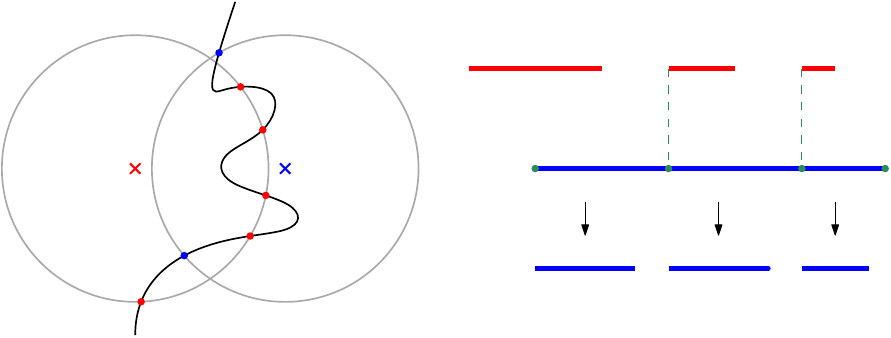}
    \caption{Example of the interval update procedure for overlapping disks. The intervals in $I_{i-1}$ are depicted in red and the interval in $I_{i}$ is depicted in blue. It is subdivided into three  intervals based on the starting points of the overlapping red intervals. Subsequently, the right borders are updated to the right borders of the respective red intervals plus the allowed time difference.}
    \label{fig:overlap}
\end{figure}

While  the  update procedure  for non-overlapping disks can only decrease the total number of intervals, the one that takes care of overlaps might also increase it. However, as the starting points of the new intervals are inherited from starting points of intervals in $I_{i-1}$,  there can be at most $kn$ interval starting points per disk during the course of the algorithm. While multiple intervals from $I_{i-1}$ might be considered for a single interval $[a,b]$ from $I_i$, the set of intervals from $I_i$ strictly contained $[a,b]$ it is only considered for $[a,b]$ and not for any other interval from $I_i$ as those are all pairwise non-overlapping. Therefore, apart from the binary searches, the number of interval access operations is bounded by $|I_{i-1}|+|I_{i}| \in \mathcal{O}(kn)$. Thus, the total running time per sweep is in $\mathcal{O}(k^2 n \log{nk})$.
\begin{lemma}\label{lem:overlap}
    COMMA runs in $\mathcal{O}(k^2  n \log{nk})$ on path graphs. 
\end{lemma}
\section{Improved Running Times on Realistic Inputs}\label{sec:intersection}
The running time derived in Lemma  \ref{lem:overlap} crucially depends on the number of intervals that are created by intersecting the path segments with the disk boundaries.  In the worst-case analysis conducted above, we assume that (almost) all disks overlap  each other and that each disk boundary is intersected by $\Theta(n)$ segments. But this is a highly unlikely scenario to occur in practice and thus the respective running times might not reflect the behavior of our algorithm on real-world inputs.
For a more realistic model, we assume in this section that each disk only overlaps a constant number of other disks and that each disk boundary is intersected by at most $\sigma$ segments. While $\sigma$ is also expected to be small in real inputs, we  use it as a parameter in our analysis to make the dependency on this value transparent.

Given this model, the total number of intervals that COMMA has to deal with is bounded by $\mathcal{O}(k\sigma)$. Accordingly, the running time of the sweeps, which can be expressed as  $\mathcal{O}(\sum_{i=1}^k \sigma \log \sigma)$, reduces  to $\mathcal{O}(k \sigma \log \sigma)$. Combined with  the interval computation, which takes $\mathcal{O}(kn)$  as discussed above, we get a total running time of $\mathcal{O}(kn + k\sigma \log \sigma)$. Thus, for $\sigma \in o(n)$, the interval computation phase now dominates the overall running time. To improve on that,  we need a faster way to compute the intersections of path segments with  the disk boundaries and to infer the respective intervals. So our goal is now to construct suitable data structures for these tasks.
\subsection{A segment-circle intersection data structure for paths}
In \cite{gupta1994intersection}, it was discussed how to preprocess a given set of $n$ line segments in $\mathbb{R}^2$ such that for a query circle $C$ of radius $r$ the subset of segments that intersect $C$ can be reported efficiently.
There are two types of segment-circle intersections (both visible in Figure \ref{fig:sweep}):
\begin{itemize}
    \item[(i)] The segment intersects $C$ once. Here, one segment endpoint is inside $C$ and the other one is outside of $C$.
    \item [(ii)] The segment intersects $C$ twice. Here, both segment endpoints are outside of $C$ but the segment point closest to the center of $C$ is inside $C$.
\end{itemize}
It is suggested in \cite{gupta1994intersection} to build a separate data structure for each type. 
We review and analyze these data structures in more detail below. For the first intersection type, we  propose a significantly faster and simpler data structure for our use case in which the segments are known to form a path.  For the second intersection type, we show that the method proposed in \cite{gupta1994intersection} only works for segments that are longer than the circle diameter. We then describe how to handle the missing case without increasing the preprocessing or query time. This result applies to arbitrary segment sets and might be of independent interest.

Subsequently, we describe in more detail how to leverage the resulting data structures for improving the COMMA running time.

\subsubsection{Dealing with intersections of type (i)}
To detect segments that have a single intersection point with a given query circle $C$, a multi-step preprocessing is applied to the segment set $S$ in \cite{gupta1994intersection}. It relies on a geometric transformation in which points $p=(x,y) \in \mathbb{R}^2$ are mapped to points $p' = (x,y,x^2+y^2) \in \mathbb{R}^3$ by a projection function $\psi(p) = p'$, and a circle $C$ with center $(a,b) \in \mathbb{R}$ and radius $r$ is mapped to the plane $z= a (2x-a) + b (2y-a) + r^2$ in $\mathbb{R}^3$ by a function $\rho(C) = z$ such that $p$ is inside $C$ iff $\psi(p)$ lies below $\rho(C)$.
We now  recap the  steps of the preprocessing algorithm.
\begin{enumerate}
    \item Construct a spanning path $\Pi$ on the $2n$ endpoints of the segments in $S$ with stabbing number $\sigma$, that is, every potential query circle is stabbed by at most $\sigma$ segments of $\Pi$. For any input point set, there exists such a spanning path with a stabbing number $\sigma \in \mathcal{O}(\sqrt n)$. The respective construction takes $\mathcal{O}(n \sqrt n \log n)$ \cite{chazelle1989quasi}.
    \item Compute an $s$-tree $T$ on $\Pi$. $T$ is a balanced binary tree in which the root represents the whole path and its children are constructed recursively by dividing the path into two subpaths of roughly equal size. The leaf nodes coincide with single path segments. $T$ can be constructed in $\mathcal{O}(n \log n)$.
    \item Let $T_v$ denote the subtree of $T$ rooted at some node $v \in T$ and let $\Pi_v$ be the associated subpath of $\Pi$. Let $\Psi_v := \{\psi(p)|p \in \Pi_v\}$. For all $v \in T$, associate with $v$ the convex hull $CH(\Psi_v)$. Convex hull computation in three dimensions can be accomplished in $\mathcal{O}(n \log n)$. Thus, the total time for this step is in $\mathcal{O}(n \log^2 n)$.
    \item Let $\overline{\Pi}_v := \{\psi(y) | \overline{xy} \in S \vee x \in \Pi_v\}$ be the set of projected other endpoints of segments with one endpoint in $\Pi_v$. Construct a halfspace reporting data structure $H_v$ on $\overline{\Pi}_v$ for each $v \in T$. This process also takes $\mathcal{O}(n \log^2 n)$.
\end{enumerate}
The overall preprocessing time is dominated by the first step and thus takes $\mathcal{O}(n \sqrt n \log n)$. 
Given a query circle $C$, one first identifies the set of canonical nodes in $T$. These are the highest nodes $v \in T$ for which $CH(\Psi_v)$ is not stabbed by the plane  $\rho(C)$. Accordingly,  for a canonical node $v$, all points in $\Psi_v$ are either all below $\rho(C)$ or all above $\rho(C)$. It was proven in \cite{chazelle1989quasi} that the number of canonical nodes is in   $\mathcal{O}(\sigma \log n)$ and that they can be detected with a matching running time by traversing the s-tree using DFS. For a segment $\overline{xy} \in S$ to have an intersection of type (i) with $C$, we know that w.l.o.g. $x$ must be inside $C$ and $y$ outside $C$. By the properties of the applied geometric transformations, this is equivalent to the requirement that $\psi(x)$ is below $\rho(C)$ and $\psi(y)$ is above $\rho(C)$. Thus, if w.l.o.g. all points in $\Psi_v$ are below $\rho(C)$, we ask the halfspace reporting data structure $H_v$ associated with $v$ for the subset of points in $\overline{\Pi}_v$ that are above $\rho(C)$. Each point in this result set is now guaranteed to be the endpoint of a segment that has a type (i) intersection with $C$. The halfspace query takes $\mathcal{O}(\log n + q')$ time where $q'$ is the output size. Performing the query for all canonical nodes and aggregating their outputs  results in a query time of $\mathcal{O}(\sigma \log^2n + q)$ where $q$ is the total number of segments that have a type (i) intersection with $C$.

We could simply use this approach for our application. However, we have the additional knowledge that the segments in our input form a path. The hope is to utilize this knowledge to improve the performance. As the  stabbing number of our input path is assumed to be  $\sigma$ with respect to the query circles we are interested in, we can fully skip the first (and most costly) step of the above described preprocessing algorithm. In addition, we can omit the fourth step and significantly simplify the query algorithm as follows: We also query the $s$-tree with $\rho(C)$. But instead of the canonical nodes, we identify the leaf nodes for which $\rho(C)$ intersects $CH(\Psi_v)$. These can be identified by using DFS and not exploring subtrees for which $\rho(C) \cap CH(\Psi_v) = \emptyset$. The segments associated with these leaf nodes are exactly the segments with a type (i) intersection with $C$. 
If there are $q$ such segments, DFS traverses $\mathcal{O}(q \log n)$ tree nodes to identify them. The intersection check of $\rho(C)$ with $CH(\Psi_v)$  can be performed in $\mathcal{O}(\log n)$ time, resulting in an overall query time of $\mathcal{O}(q \log^2 n)$.  The following lemma summarizes our findings.
\begin{lemma}
    A path with $n$ segments can be preprocessed in $\mathcal{O}(n \log^2 n)$ into a data structure that reports the $q$ segments that have a type (i) intersection with a query circle in $\mathcal{O}(q \log^2n)$.
\end{lemma}
We emphasize that the performance of our data structure tailored to paths no longer explicitly depends on  $\sigma$. In the COMMA analysis on inputs adhering to our proposed realistic model we will however assume $q \leq \sigma$ for all disk boundaries.

\subsubsection{Dealing with intersections of type (ii)}
Segments with two intersection points with  $C$ are not detected by the data structure described above, as here both segment endpoints are outside of $C$ and hence projected above $\rho(C)$. Thus, their convex hull does not induce an intersection with said plane. Therefore, another data structure was used in \cite{gupta1994intersection} to detect segments with a type (ii) intersection based on the following characterization: For a segment $s=\overline{xy}$ to intersect $C$ twice, the center of  $C$ has to lie in the so called truncated strip of $s$. The strip is formed by the two lines through $x$ and $y$ that are perpendicular to $s$. The truncation happens at distance $r$ from $s$, resulting in a rectangle $R_s$ of length $|s|$ and height $2r$ centered at $s$. Based thereupon, the goal is to construct a data structure that preprocesses a given set of rectangles such that for a query point (the circle center) all rectangles that enclose the query point can be reported efficiently. As efficient data structures are known for the case where the input is a set of triangles, \cite{gupta1994intersection}  suggest to simply cut each rectangle into two triangles using one of its diagonals. Figure \ref{fig:sci2} (a) and (b) illustrate these concepts. \cite{gupta1994intersection} use the triangle enclosure data structure proposed in \cite{cheng1992algorithms} with a preprocessing time in $\mathcal{O}(n\sqrt n  \log^2 n)$ and a query time in $\mathcal{O}(\sqrt n \log^2 n + q)$ where $q$ is the output size. Alternatively, \cite{overmars1990storing} describe a data structure  for the same problem with a $\mathcal{O}(n \log^3 n)$ preprocessing time and a query time in $\mathcal{O}(n^\lambda + q)$ with $\lambda \approx 0.695$.
\begin{figure}
    \centering
    \includegraphics[width=0.85\linewidth]{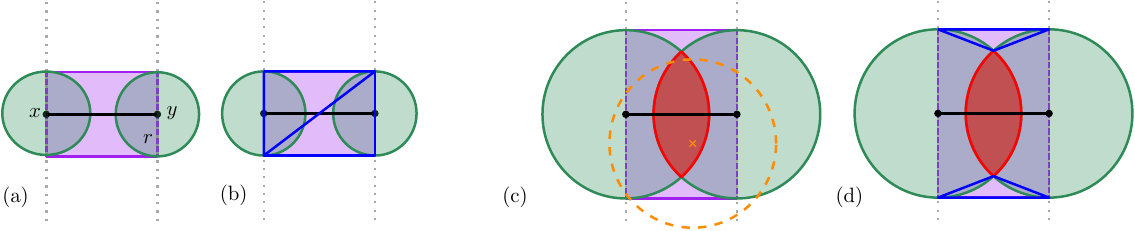}
    \caption{(a) Example segment $s=\overline{xy}$ (black) longer than $2r$. $R_s$ is depicted in purple. Its subdivision into triangles is indicated by the blue lines in (b). In (c), segment $s$ is shorter than $2r$. The disks $D_x$ and $D_y$ are depicted in green. Their intersection (red zone) is problematic as query circles with a center in that zone do not intersect $s$ as indicated by the orange example circle. (d) illustrates the alternative triangle construction for this case.}
    \label{fig:sci2}
\end{figure}

 The issue with this approach is that it might  report segments that do not actually intersect $C$. If the center of $C$ is in $D_x \cap D_y$, then $C$ contains both endpoints of $s$ and consequently there is no intersection between $C$ and $s$. Thus, circles with a center in $R_s \cap D_x \cap D_y$ have to be exempt from being reported (post filtering would impair the output-sensitivity). If $2r \leq |s|$, we have $D_x \cap D_y = \emptyset$ and hence this scenario cannot occur. But if $2r > |s|$, that is, the circle diameter exceeds the segment length,  $D_x$ and $D_y$ overlap inside $R$. Figure \ref{fig:sci2} (c) depicts this scenario. This issue can be fixed by constructing two triangles which connect the base sides of  $R_s$ to the  intersection points of $D_x$ and $D_y$ which are furthest from $s$, respectively. The triangles are  shown  in Figure \ref{fig:sci2} (d). By construction they do not contain any point in $D_x \cap D_y$ and  they fully cover the relevant area $R_s  \setminus (D_x \cup D_y)$. Therefore, we can construct a triangle enclosure data structure with two triangles per segment as before, but with a different triangle construction mechanism depending on the segment length and $r$.

 Putting together the data structures for type (i) and type (ii) segment-circle intersections, we get the following corollary.
 \begin{corollary}\label{cor:cis}
     A set of $n$ segments can be preprocessed in time $\mathcal{O}(n \sqrt n \log^2 n)$ such that the $q$ segments that intersect a query circle of fixed radius $r$ can be computed in $\mathcal{O}(\sqrt n \log ^2 n + q)$.
 \end{corollary}
For our special case in which the input segments form a path, using the triangle enclosure data structure by \cite{overmars1990storing}  provides a different trade-off between preprocessing and query time.
\begin{lemma}\label{lem:cis}
     A path with $n$ segments can be preprocessed in $\mathcal{O}(n \log^3 n)$ into a data structure that reports the $q$ segments that have an intersection with a query circle of fixed radius $r$  in $\mathcal{O}(n^{0.695} + q \log^2n)$. 
\end{lemma}

\subsection{COMMA acceleration}
The goal is now to leverage our new segment-circle intersection data structure for accelerating the interval computation phase of COMMA. For each disk $D_i$, we query the data structure with the disk boundary circle $C_i$ to retrieve the set of segments that intersect $C_i$. Then, we compute the explicit intersection points for each segment. To get the final intervals, we also need the travel time from the start point of $P$ to the respective intersection point.  To have  access to this value in constant time, we also compute a prefix sum array for the edge costs along $P$. Querying this data structure with an edge segment index and a point on that segment only demands to interpolate the cost assigned to that segment in the array up to the given point. The resulting intersection point timestamps are sorted to construct the respective intervals.  Using this procedure, we can compute the set of intervals for each $D_i$ without having to traverse any edges that are fully contained in $D_i$. Note that the DAG based algorithm always has to consider the endpoints of all of these segments to identify the candidate location sets. Thus, the running time of COMMA is significantly reduced while the $\mathcal{O}(k n^2)$ running time of the baseline remains unchanged in our realistic model.
\begin{theorem}
    The running time of COMMA on inputs adhering to the realistic model is in  $\mathcal{O}(k n^{0.695} + n \log^3 n +  k \sigma \log^2 n)$.
\end{theorem}
\begin{proof}
    The prefix sum array can be computed in $\mathcal{O}(n)$ and the segment-circle intersection data structure in $\mathcal{O}(n \log^3 n)$.
    The time to compute the boundary intersection points for  $D_i$ is in $\mathcal{O}(n^\lambda + q_i \log^2 n) $ for $\lambda=0.695$ where $q_i$ denotes the number of path segments intersecting the boundary of $D_i$. Thus, the total time for intersection queries is $\sum_{i=1}^k n^\lambda +  q_i \log^2 n = k n^\lambda + \log^2 n \sum_{i=1}^k q_i$. 
    As we have $|q_i| \leq \sigma$, it follows that $\sum_{i=1}^k q_i \in \mathcal{O}(k \sigma)$.  With the help of the prefix sum array, the final interval computation takes $\mathcal{O}(k \sigma \log \sigma)$, including sorting of the intersection points by timestamp. In total, the interval computation phase takes $\mathcal{O}(k n^\lambda + k \sigma \log^2 n + n \log^3 n)$ time. The running time for the  sweeps is in $\mathcal{O}(k \sigma \log \sigma)$. Thus, the overall running time is  dominated by the interval computation phase.
\end{proof}
Alternatively, using Corollary \ref{cor:cis} instead of Lemma \ref{lem:cis}, the running time of COMMA is in $\mathcal{O}(k\sqrt n \log^2n + n\sqrt n \log^2n + k \sigma \log \sigma)$. These results allow us to explore the role of $\sigma$ for the overall running time of COMMA. In practice, $\sigma$ is expected to be a small constant. However, we observe that as long as $\sigma \in \mathcal{O}(\sqrt n)$ the COMMA running time is in $\mathcal{O}(k n^{0.695} + n \log^3 n)$ or $\mathcal{O}(k\sqrt n \log^2n + n\sqrt n \log^2n)$, respectively. So even for $k \in \Theta(n)$, we  now have  a running time that is clearly subquadratic in $n$ while the running time  of the baseline would be cubic for such a $k$, even in the realistic model.

\section{Experimental Results}

We implemented COMMA as well as the DAG baseline algorithm in C++. Experiments were conducted on a single core of an AMD Ryzen Threadripper 3970X.

\subsection{Benchmark data}
In the evaluation, we use two types of data sets: generated paths and measurement sequences as well as  GTFS data. The former is created as follows. To get a  path $P$ of length $n$, we add vectors with coordinates that are uniformly sampled at random. The time values for the points are accumulated in the same way. The corresponding measurement sequence $M$ is generated as follows. We first randomly sample $k$ values in $[0,1]$ and sort them ascendingly. These values are used to find positions along the path with an associated time value. Perturbing said points gives us our generated input. Figure \ref{fig:MinimalExampleCommaBaseline} shows a small example and Figure \ref{fig:commaExampleBig} a large one.

For real data, we used GTFS bus timetables. We extract the  coordinates present in \textit{shapes.txt} of  a bus line, and its corresponding stops in \textit{stops.txt} (defining a trip of the vehicle). The time values for the trip are present in \textit{stop\_times.txt}. Using said times together with the \textit{shape\_dist\_traveled} data in \textit{stop\_times.txt}, we can induce valid time values for the shape of the bus line. 
An example is shown in Figure \ref{fig:Olbia}.

\begin{figure}[h!]
    \centering
    \includegraphics[width=0.3\textwidth]{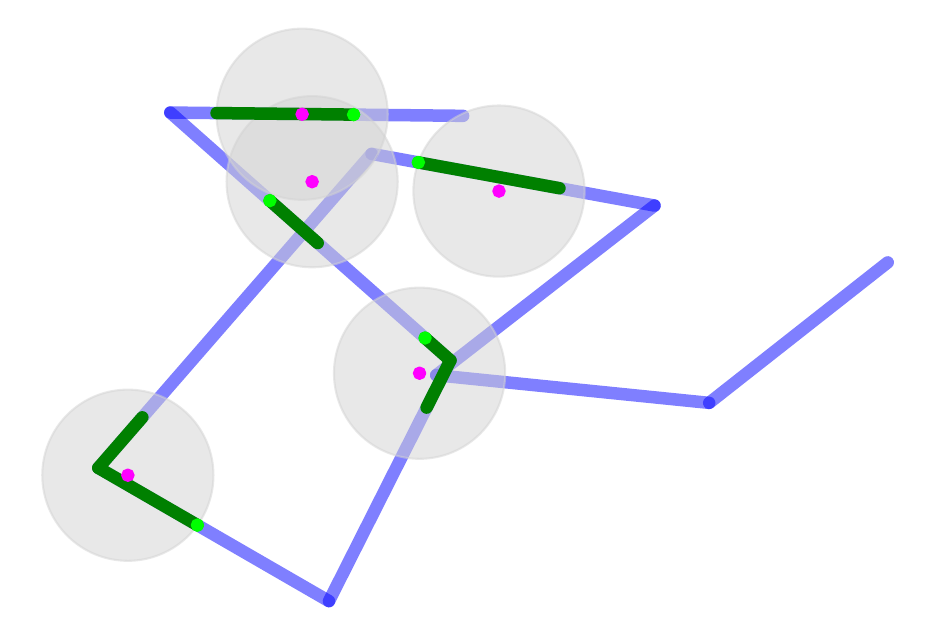}
    \hspace{1cm}
    \includegraphics[width=0.3\textwidth]{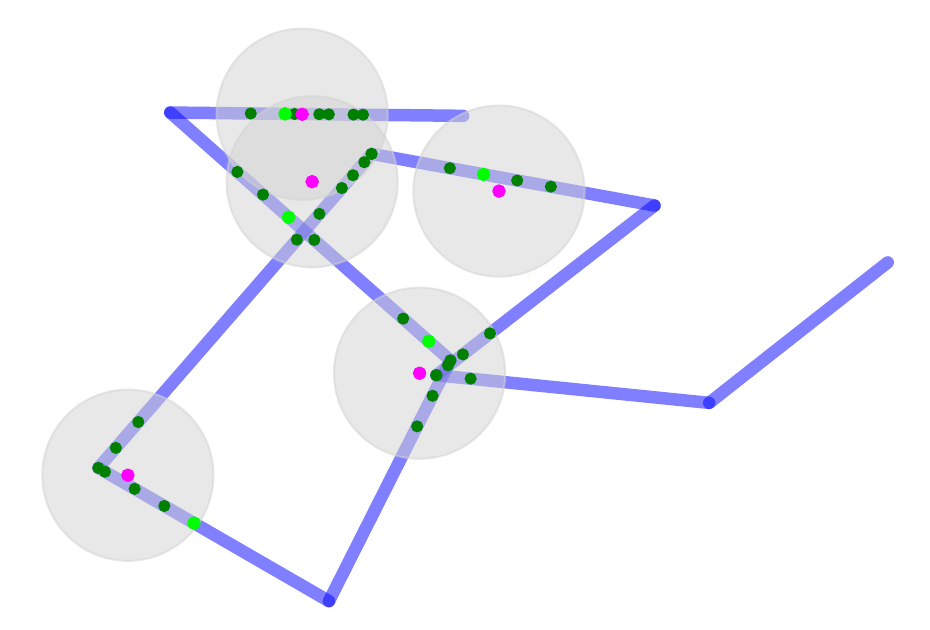}
    \caption{Exemplary execution of COMMA and the baseline for $n = 10, k = 5$ and $r = 0.5$. The sequence points and the mapped points are depicted in magenta and light green respectively. Shown in green are the intervals of COMMA and the candidate points of the baseline. }
    \label{fig:MinimalExampleCommaBaseline}
\end{figure}

\begin{figure}[h!]
    \centering
    \includegraphics[width=\textwidth]{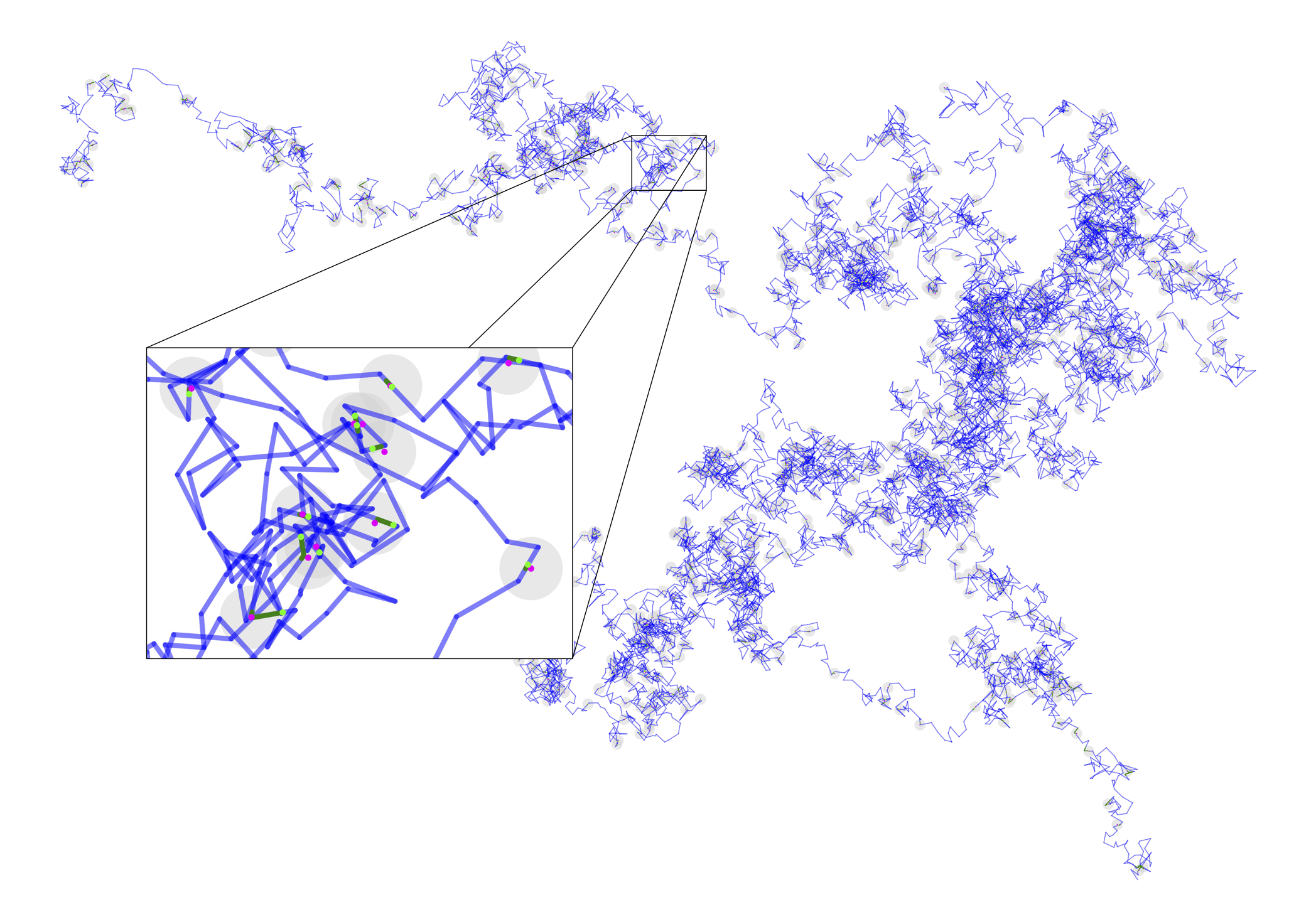}
    \caption{Execution of COMMA for $n=10000$, $k=1000$ and $r=1$. The sequence points and the mapped points are depicted in magenta and light green respectively. Shown in green are the intervals of COMMA.}
    \label{fig:commaExampleBig}
\end{figure}

\subsection{Segment-circle intersection data structure}
For the DAG algorithm as well as for COMMA, we need to retrieve the intervals of $P$ inside the disks $D_1, \dots, D_k$, and for the DAG algorithm additionally the contained segment endpoints as those serve as candidate locations. However, it might happen that the DAG algorithm does not find a feasible solution when restricted to segment endpoints inside the disks. Thus, we also implement a sampling variant, which additionally places possible locations along the edges inside $D_i$ equidistantly (see again Figure \ref{fig:MinimalExampleCommaBaseline}). This allows us to study what sampling rates are needed to find feasible mapping sequences.

Naively, the intervals and segment endpoints can be retrieved in time $\mathcal{O}(kn)$ for both algorithms. But as discussed in our theoretical analysis, interval computation might easily dominate the running time of COMMA and thus using a suitable segment-circle intersection data structure is advised. However, the data structures we used in the theoretical analysis are quite involved, using  transformations to three dimensions as well as  several nested data structures (especially for intersection of type (ii)) \cite{gupta1994intersection}. To the best of our knowledge, no library implementation is available. Engineering these data structures would certainly be interesting, but a naive implementation is expected to be rather slow. Thus,  we decided to use  CGAL's AABB Tree in both algorithms, DAG and COMMA. It is a hierarchical data structure where each node represents an axis-aligned bounding box (AABB) that encloses the set of geometric primitives that are represented by the respective leaf nodes. In our case, the leaf nodes store the path segments. Thus, the search algorithm is similar to the one we described in Section \ref{sec:intersection} only that here the oracle that decides whether to explore a branch of the tree might be false positive in case only the bounding box of the segment and the disk intersect but not the elements themselves.  But it allows us to handle both possible types of segment-circle intersection in one data structure as well as segment endpoint extraction for the DAG baseline. 

\begin{figure}
    \centering
    \includegraphics[width=0.48\textwidth]{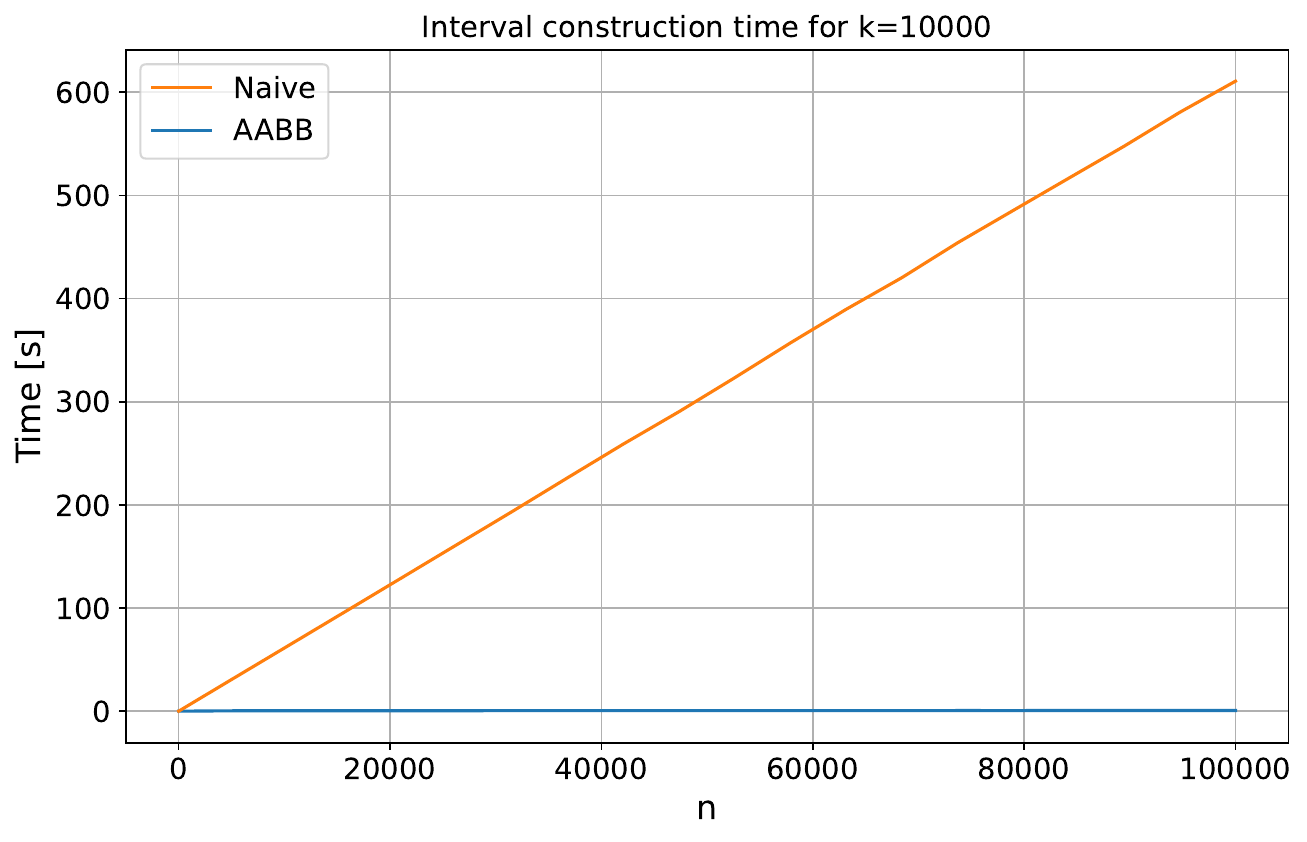}
    \hfill
    \includegraphics[width=0.48\textwidth]{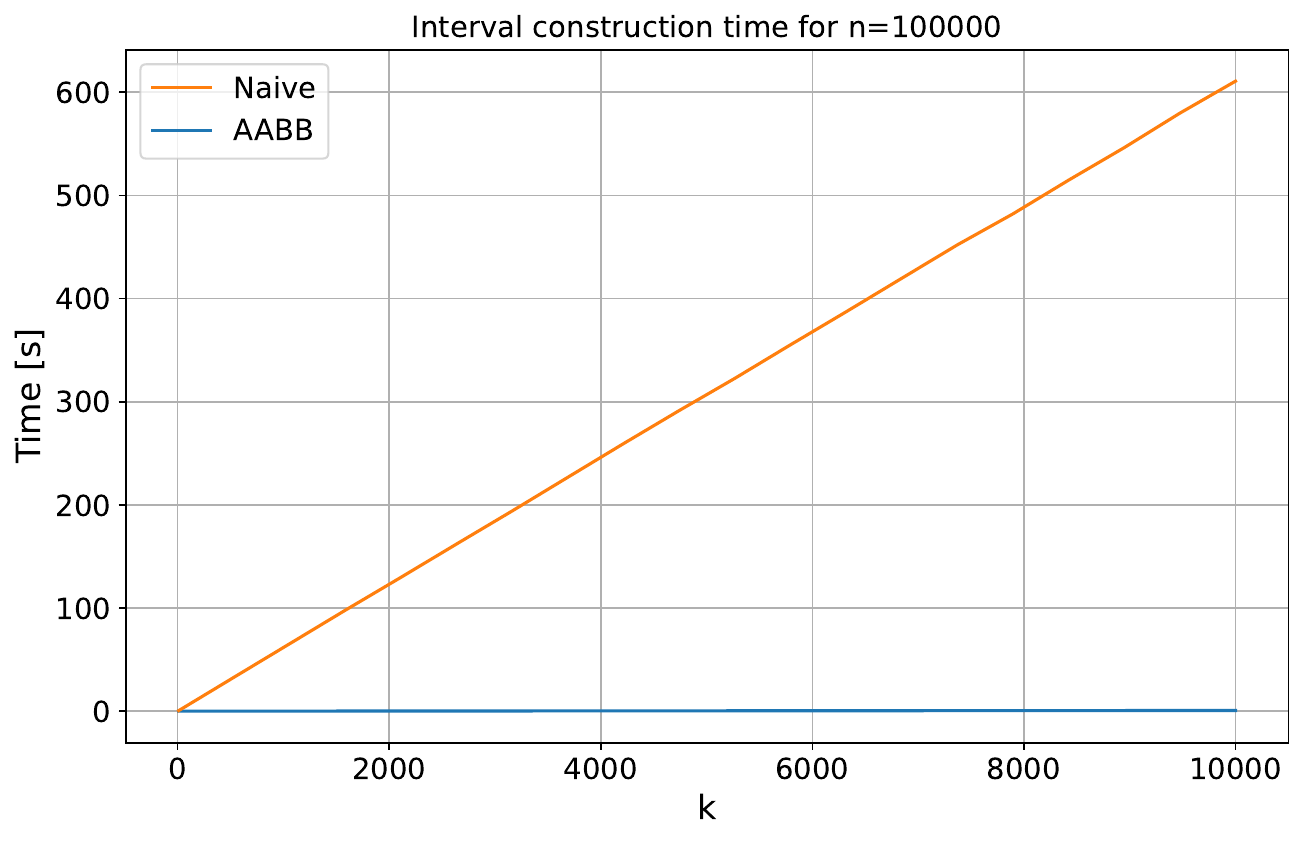}
    \caption{Naive interval construction times  and times using the AABB Tree. The left plot fixes the number of measurements $k$ while the right plot does the same for the length $n$ of $P$.}
    \label{fig:IntervalComp}
\end{figure}

In Figure \ref{fig:IntervalComp}, we evaluate the running time for interval and segment extraction on generated paths with varying values of $n$ and $k$. 
We observe that the running time of the naive approach  increases linear with both $n$ and $k$ and is in the order of minutes already for small parameter values, while the AABB Tree provides results very quickly for all choices of $n$ and $k$.

\subsection{COMMA versus DAG}
\begin{figure}
    \centering
    \includegraphics[width=0.48\textwidth]{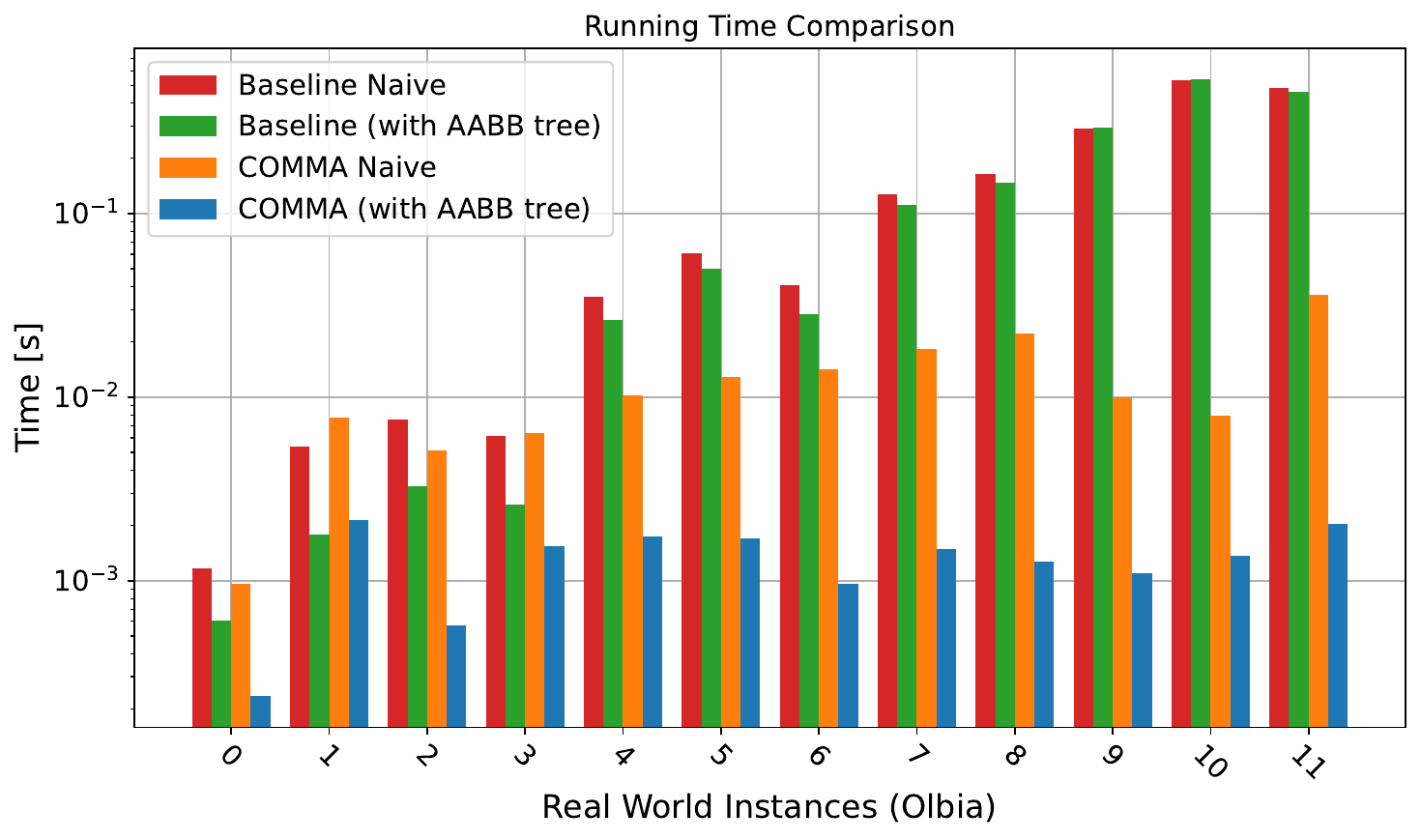}
    \hfill
    \includegraphics[width=0.48\textwidth]{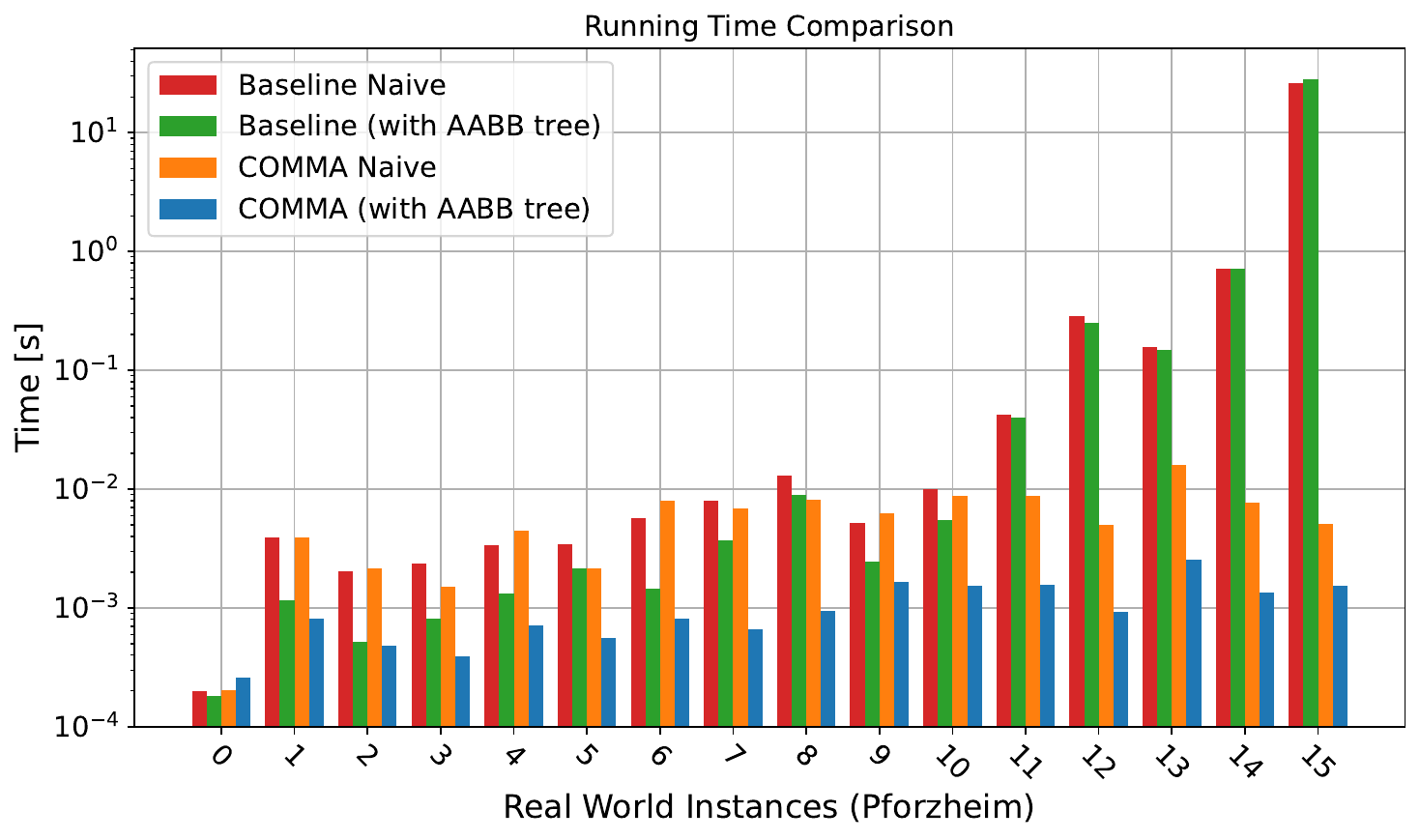}
    \caption{Running time comparison of the COMMA baseline variants for 12 and 16 real world GTFS instances for Olbia and Pforzheim, with average instance sizes of $n \approx 704, k \approx 42$ and $n \approx 634, k \approx 16$ respectively. The algorithms were run for the smallest radius that yielded a solution for COMMA and, given this radius, for the first feasible sampling distance of the baseline.}
    \label{fig:RWD}
\end{figure}
Next, we compare COMMA to the DAG baseline on real and generated input data. In Figure \ref{fig:RWD}, a comparison of the algorithms on  GTFS bus stop mapping instances is provided. For the baseline, the running time difference between using or not using the AABB tree is rather small, as most of the effort is required to construct the DAG. We remark that the unsuccessful runs of the baseline due to a too low sampling rate to identify a feasible match are not even included in the timings.  For COMMA, the running time is already better on the larger instances without the AABB tree, and with the help of the AABB tree it achieves speed-ups of up to four orders of magnitude. This allows for much faster GTFS data processing especially for data sets with thousands of routes. 

\begin{figure}
    \centering
    \includegraphics[width=0.5\textwidth]{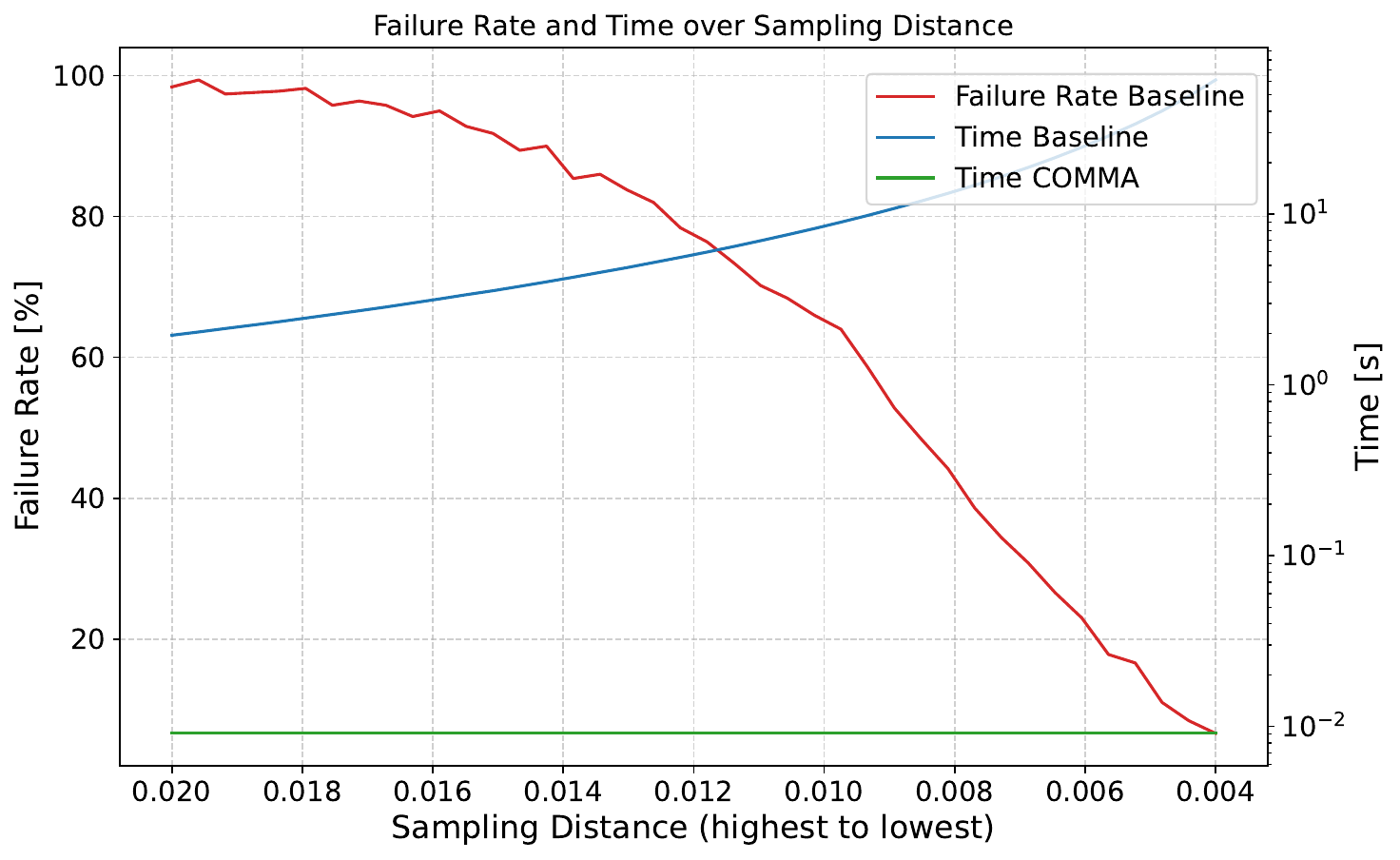}
    \hfill
    \includegraphics[width=0.47\textwidth]{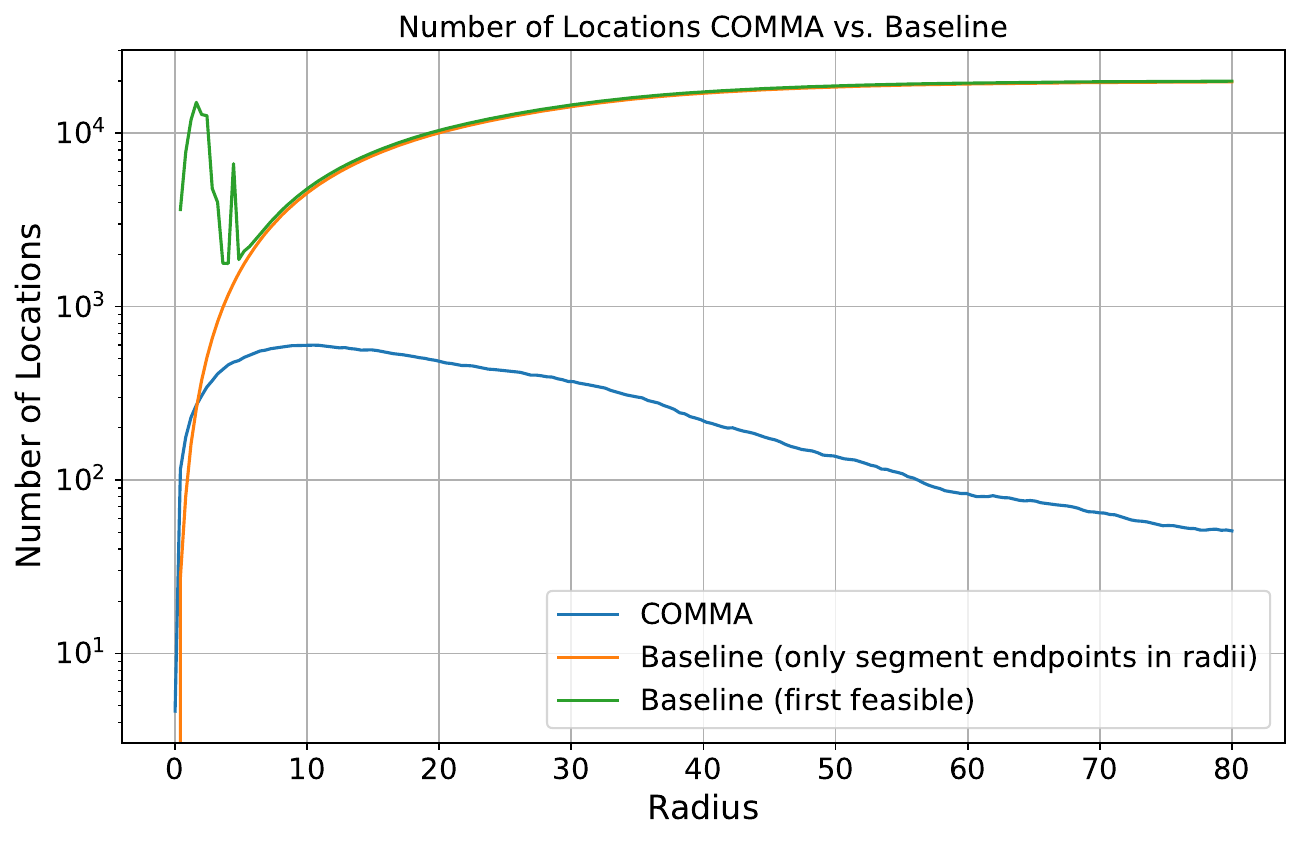}
    \caption{Depicted on the left are the failure rate of the baseline algorithm as well as the time of the baseline algorithm and COMMA over sampling distances varying from 0.02 to 0.004 for 500 instances of $n=1000$ and $k=100$. The COMMA runtime is constant as it is independent of the sampling distance. Shown on the right is the number of locations that are considered by the algorithms depending on the radius for 50 instances of size $n = 1000$ and $k = 20$ that is varied from 0.01 to 80. For COMMA these locations are the number of endpoints of the intervals, and for the baseline the number of sampling candidates.}
    \label{fig:RadiusSampleDist}
\end{figure}
For further scalability studies, we use the above described generator. In Figure \ref{fig:RadiusSampleDist}, left, we observe that the  DAG baseline without additional samples fails to produce a feasible solution on the example instance. With a higher sampling rate and lower sample distance, the failure rate decreases but at the same time the running time increases significantly, as the DAG gets much larger. COMMA produces a feasible mapping path and outperforms the DAG algorithm by four orders of magnitude. A similar behavior was observed across all generated instances. The main reason for this huge speed-up is illustrated in Figure \ref{fig:RadiusSampleDist}, right. The number of intervals that COMMA has to consider is significantly smaller than the number of locations that the DAG algorithm needs to process except for very small radii. As the DAG algorithm additionally scales worse with the number of locations to consider than COMMA, the running time gap gets more pronounced the more measurements need to processed and the larger the disks are. 

Executing the DAG baseline on the instance shown  in Figure \ref{fig:commaExampleBig} with  $n=10000$, $k=1000$ does not yield a feasible result even with a sampling distance as small as $0.004$, although the  DAG already contains $3.5$ million candidate nodes and $796$ million edges (that respect the time bound for consecutive locations). COMMA, however, scales very well and can easily identify feasible mapping paths for such instances. Indeed, for the generated instances depicted in Figure \ref{fig:IntervalComp} with up to $n=100000$ segments and up to $k=10000$ measurements, COMMA produces the solution in less than a second.  Thus, COMMA is also suitable to be used as an efficient subroutine in map matching algorithms for general graphs that select a set of candidate paths and then check whether they constitute feasible matching paths for GPS sequences with potentially thousands of measurements.

To further illustrate the scalability of COMMA, Table \ref{tab:Scalability} shows results for very large $n$ and $k$ and also provides a breakdown of the running time. Even for millions of path segments and measurements, the running time stays within a few minutes. We see that the interval computation dominates the overall running time even when using the AABB tree, while the sweeps and the path extraction are very efficient. 
\begin{table}
    \centering
    \footnotesize
    \begin{tabular}{cc||r|r|r|r|r}
        $n$ & $k$ & Intervals & Time [s] & Intervals [s] & Sweeps [s] & Extraction [s]\\
        \hline
        $10^8$ & $10^4$ & 147,888  & 182.67 & 182.52 & 0.13 & 0.02\\
        $10^8$ & $10^5$ & 1,415,391 & 256.40 & 246.07 & 9.55 & 0.17 \\
        $10^8$ & $10^6$ & 13,038,123  & 271.78 & 258.43 & 11.07 & 1.25\\
        $10^8$ & $10^7$ & 138,714,559  & 976.61 & 886.39 & 73.54 & 9.85\\
    \end{tabular}
    \vspace{0.5em}
    \caption{Scalability study for COMMA. The radius for  the instances is one and each instance returned a valid sequence.}
    \label{tab:Scalability}
\end{table}

Figure \ref{fig:primitivesTestedNAndK} provides a detailed analysis of the number of segment bounding boxes checked by the AABB tree traversal and the actual number of intervals that were identified for varying $n$ and $k$.
\begin{figure}
    \centering
    \includegraphics[width=0.48\textwidth]{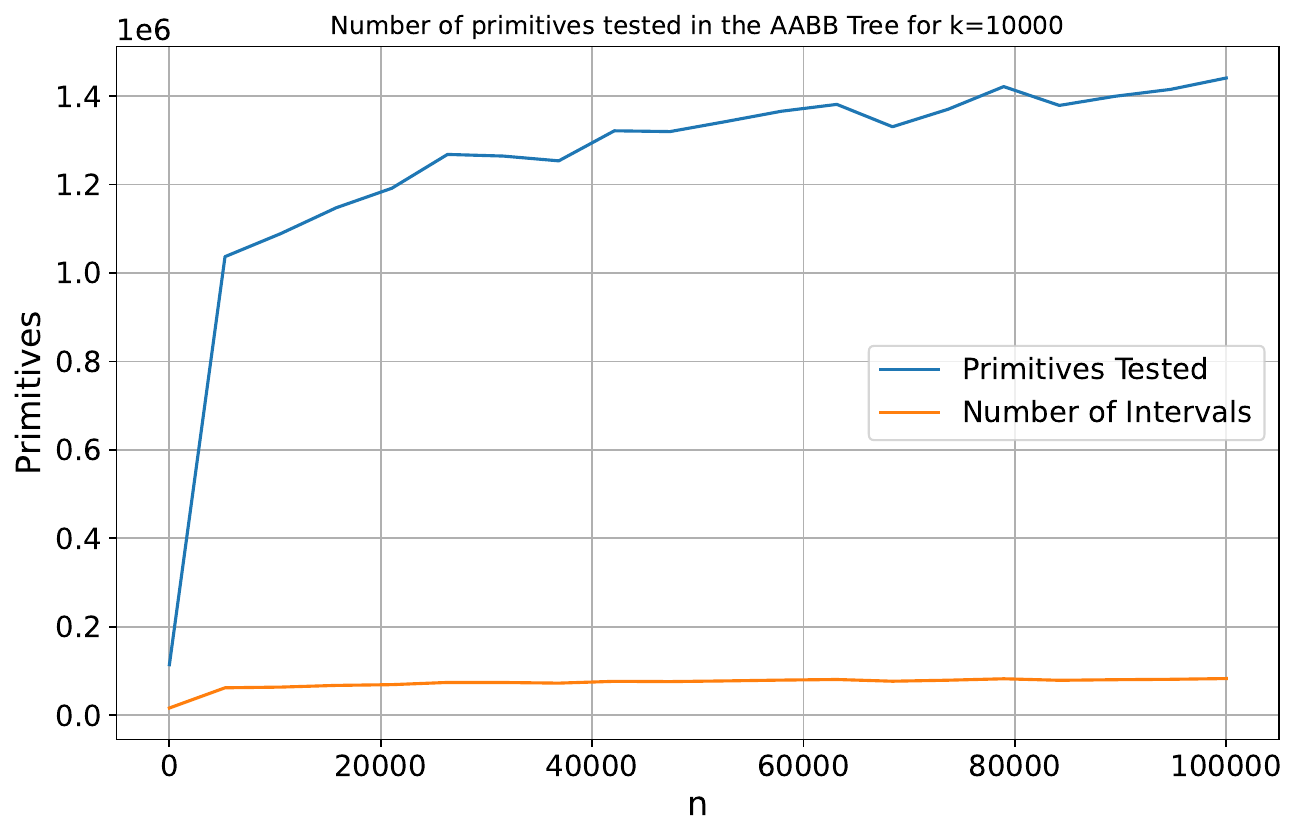}\hfill
    \includegraphics[width=0.48\textwidth]{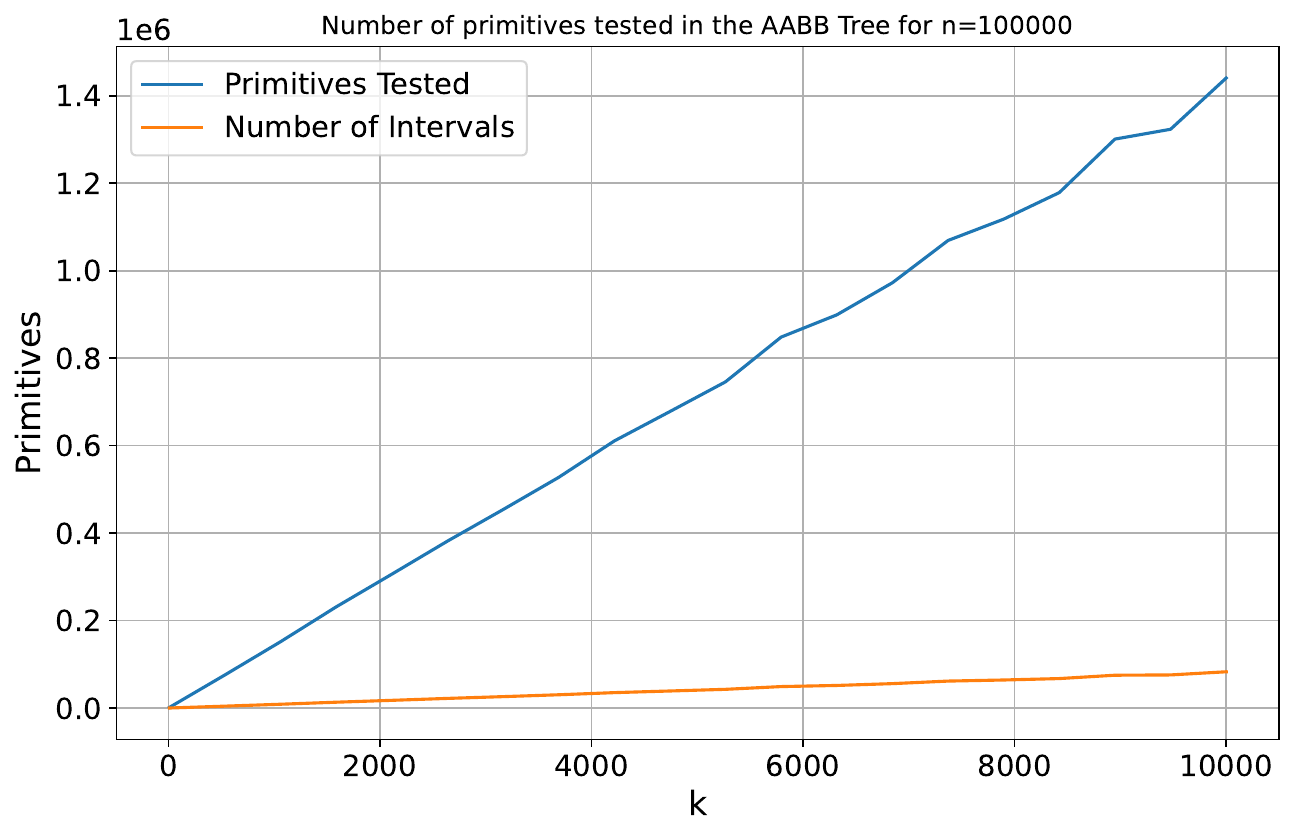}
    \caption{Number of intermediate primitives (bounding boxes) that were checked for intersection in the AABB tree during interval computation for $r=1$ and $k$ and $n$ fixed respectively.}
    \label{fig:primitivesTestedNAndK}
\end{figure}
Clearly, there is quite a significant gap between those two values. Thus, it might indeed be worthwhile to implement and engineer a  segment-circle intersection data structure with better theoretical guarantees in future work to cater for huge instances.
However, for typical instances derived from GTFS data or from GPS measurements with up to a few thousand segments and measurements, leveraging the AABB tree is sufficient to achieve very low query times.

\section{Conclusions and Future Work}
In this paper, we analyzed the weaknesses of existing map matching approaches which use discrete candidate location sets for the measurements. Our continuous approach is not only significantly faster, but more importantly, it is guaranteed to find a feasible match when there is one. The obvious next step is to extend COMMA to general graphs. This requires to deal with more complicated structures within the disks and to efficiently compute shortest paths  between disk boundary intersection points. Another direction for future work would be to leverage the fact that the intervals computed by COMMA encode all feasible mapping sequences. This allows to incorporate secondary objective functions to select a solution, for example,  distance of snapped location to the original measurements.

%%
%% Bibliography
%%
 
\bibliography{references}

\end{document}